\documentclass{IEEEtran}
\IEEEoverridecommandlockouts

\makeatletter
\def\endthebibliography{%
  \def\@noitemerr{\@latex@warning{Empty `thebibliography' environment}}%
  \endlist
}
\makeatother

\usepackage[cmex10]{amsmath}
\usepackage{amssymb,amsthm}
\DeclareMathOperator*{\argmax}{arg\,max}
\DeclareMathOperator*{\argmin}{arg\,min}
\usepackage{bm}
\usepackage{amsfonts}
\usepackage{relsize}
\usepackage{cuted}
\usepackage{subfig}
\usepackage{graphicx}
\usepackage{epstopdf}
\usepackage{algorithm}
\usepackage[noend]{algpseudocode}
\usepackage{bbm}
\usepackage{comment}
\usepackage{multirow}
\usepackage{pbox}
\usepackage{array}
\usepackage{booktabs}
\usepackage{multicol}
\usepackage{url}
\usepackage{hyperref}
\usepackage{cleveref}

\newtheorem{lemma}{Lemma}

\newcommand*\diff{\mathop{}\!\mathrm{d}}

\usepackage{textcomp}
\usepackage{cite}
\usepackage{subfig}
\usepackage{xcolor}
\def\BibTeX{{\rm B\kern-.05em{\sc i\kern-.025em b}\kern-.08em
    T\kern-.1667em\lower.7ex\hbox{E}\kern-.125emX}}
\usepackage{enumitem} 
\usepackage[font=small,labelfont=bf]{caption}
\usepackage{capt-of}

\algnewcommand{\Inputs}[1]{%
  \State \textbf{Inputs:}
  \Statex \hspace*{\algorithmicindent}\parbox[t]{.8\linewidth}{\raggedright #1}
}
\algnewcommand{\Initialize}[1]{%
  \State \textbf{Initialize:}
  \Statex \hspace*{\algorithmicindent}\parbox[t]{.8\linewidth}{\raggedright #1}
}

\makeatletter
\def\BState{\State\hskip-\ALG@thistlm}
\makeatother
\title{MIND: Maximum Mutual Information Based Neural Decoder}

\author{\IEEEauthorblockN{Andrea M. Tonello and Nunzio A. Letizia} 
\thanks{The authors are with Universit\"{a}t Klagenfurt, Institute of Networked and Embedded Systems, Klagenfurt 9020, Austria, e-mail: andrea.tonello@aau.at, nunzio.letizia@aau.at.}
}


\IEEEoverridecommandlockouts

\begin{document}

\maketitle


\begin{abstract}
We are assisting at a growing interest in the development of learning architectures with application to digital communication systems. Herein, we consider the detection/decoding problem. We aim at developing an optimal neural architecture for such a task. The definition of the optimal criterion is a fundamental step. We propose to use the mutual information (MI) of the channel input-output signal pair, which yields to the minimization of the a-posteriori information of the transmitted codeword given the communication channel output observation. The computation of the a-posteriori information is a formidable task, and for the majority of channels it is unknown. Therefore, it has to be learned. For such an objective, we propose a novel neural estimator based on a discriminative formulation. This leads to the derivation of the mutual information neural decoder (MIND). The developed neural architecture is capable not only to solve the decoding problem in unknown channels, but also to return an estimate of the average MI achieved with the coding scheme, as well as the decoding error probability. Several numerical results are reported and compared with maximum a-posteriori and maximum likelihood decoding strategies.  
\end{abstract}


\section{Introduction}
\label{sec:introduction}
In digital communication systems, data detection and decoding are fundamental tasks implemented at the receiver. The maximum a-posteriori (MAP) decoding approach is the optimal one to minimize the symbol or sequence error probability \cite{Proakis2001,Bahl1974}. If the channel model is known, the MAP decoder can be realized according to known algorithms, e.g., BCJR algorithm \cite{Bahl1974}. For instance, for a linear time invariant channel (LTI) with additive Gaussian noise, the decoder comprises a first stage where the channel impulse response is estimated. Then, a max-log-MAP sequence estimator algorithm is implemented. The decoding metric is essentially related to the Euclidean distance between the received and hypothesized transmitted message filtered by the channel impulse response \cite{Proakis2001}. The task becomes more intriguing for channels that cannot be easily modeled or that are even unknown and for which only data samples are available. In such a case, a learning strategy turns out to be an attractive solution and it can be enabled by recent advances in machine learning for communications \cite{Oshea2017, Nachmani2018, Dorner2018}.

In this paper, we derive a neural architecture for optimal decoding in an unknown channel. We propose to exploit the mutual information (MI) of the input-output channel pair as decoding metric. The computation of the MI is a formidable task and it is unknown in many instances. Therefore, the MI has to be learned, which can be done in principle with neural architectures \cite{Mine2018, Wunder2019, Letizia2021, Song2020, LetiziaNIPS}. 
However, they are not directly deployable since in the problem at hand the decoder has a more specific task: it has to learn and minimize the a-posteriori information $-\log_2{(p_{X|Y}(\mathbf{x}|\mathbf{y}))}$ of the transmitted codeword $\mathbf{x}$ given the observed/received signal vector $\mathbf{y}$. The direct estimation of the a-posteriori information becomes then the objective. Such an estimation is possible by directly learning the conditional probability density function (pdf) $p_{X|Y}(\mathbf{x}|\mathbf{y})$ with a new estimator that exploits a discriminative model referred to as mutual information neural decoder (MIND). MIND allows not only to perform the decoding task, but also to return an estimate of the channel achievable rate and of the decoding error probability. 

This paper is organized as follows. Section \ref{sec:mid} presents the decoding approach principle. Section \ref{sec:ndec} describes the core idea about MIND. Section \ref{sec:implementation} discusses two practical neural architectures. Section \ref{sec:estimation} explains how MIND can be exploited to estimate both the achievable rate and the decoding error probability. Numerical results for different scenarios are reported in Section \ref{sec:results}. Finally, the conclusions are drawn.


\section{Mutual Information Decoding Principle}
\label{sec:mid}

The considered communication system comprises an encoder, a channel, and a decoder. The encoder maps a message of $k$ bits into one of the $M=2^k$ encoded messages of $n$ (not necessarily binary) symbols. The uncoded message is denoted with the vector $\mathbf{b}$, while the coded message with the vector $\mathbf{x}$ belonging to the set $\mathcal{A}_x$. Thus, the code rate is $R=k/n$ bits per channel use. The channel conveys the coded message to the receiver. Without loss of generality, we can write that the received message is $\mathbf{y}=H(\mathbf{x},\mathbf{h},\mathbf{n})$, where the channel transfer function $H$ implicitly models the channel internal state $\mathbf{h}$ including stochastic variables, e.g. noise, $\mathbf{n}$. For instance, for an LTI with additive noise, we obtain $\mathbf{y}=\mathbf{h}*\mathbf{x}+\mathbf{n}$, where $*$ denotes convolution.

The aim of the decoder is to retrieve the transmitted message so that to minimize the decoding error, or equivalently to maximize the received information for each transmitted message. Indeed, if we consider the mutual information of the pair $(\mathbf{y},\mathbf{x})$, with joint pdf $p_{XY}(\mathbf{x},\mathbf{y})$, and marginals $p_X(\mathbf{x})$ and $p_Y(\mathbf{y})$, this can be written as

\begin{equation}
i(\mathbf{y};\mathbf{x})=i(\mathbf{x})-i(\mathbf{x}|\mathbf{y})
\end{equation}
where $i(\mathbf{x}) = -\log_2 p_{X}(\mathbf{x})$ is the information of the transmitted message and $i(\mathbf{x}|\mathbf{y})$ is the a-posteriori information (residual uncertainty observing $\mathbf{y}$) \cite{Gallager1968}. It follows that for each transmitted message, the mutual information is maximized when the a-posteriori information is minimized. The a-posteriori information is given by
\begin{equation}
\label{eq:PI}
i(\mathbf{x}|\mathbf{y})=-\log_2 p_{X|Y}(\mathbf{x}|\mathbf{y}) = -\log_2{\frac{p_{XY}(\mathbf{x},\mathbf{y})}{p_{Y}(\mathbf{y})}}
\end{equation}
so that the decoding criterion becomes
\begin{equation}
\mathbf{\hat{x}}=\argmin_{\mathbf{x}\in \mathcal{A}_x}{i(\mathbf{x}|\mathbf{y})} = \argmax_{\mathbf{x}\in \mathcal{A}_x}{\log_2 p_{X|Y}(\mathbf{x}|\mathbf{y})}
\end{equation}
which corresponds to the log-MAP decoding principle \cite{Proakis2001,Bahl1974}.

To compute the a-posteriori information, we need to know the a-posteriori probability $p_{X|Y}(\mathbf{x}|\mathbf{y})$. If the channel transfer function can be modeled, the decoder is realized according to known approaches. In fact, the a-posteriori probability is explicitly calculated as $p_{X|Y}(\mathbf{x}|\mathbf{y})=p_{Y|X}(\mathbf{y}|\mathbf{x})p_{X}(\mathbf{x})/p_Y(\mathbf{y})$, i.e., from the conditional channel pdf, the coded message a-priori probability and the channel output pdf. However, for unknown channels and sources, we have to resort to a learning strategy as explained in the following.

\section{Discriminative Formulation of MIND}
\label{sec:ndec}

In the following, we show that it is possible to design a parametric discriminator whose output leads to the estimation of the a-posteriori information. The optimal discriminator is found by maximizing an appropriately defined value function.
Before describing the proposed solution, it should be noted that generative adversarial networks (GANs) \cite{Goodfellow2014} deploy a discriminator $D$ implemented by a neural network. It is used to distinguish if a sample is an original one coming from real data or a fake one generated by the generator $G$. The adversarial training pushes the discriminator output $D(\mathbf{a})$ towards the optimum value \cite{Goodfellow2014} $D^*(\mathbf{a}) = p_{data}(\mathbf{a})/(p_{data}(\mathbf{a})+p_{gen}(\mathbf{a}))$, 
where $p_{data}(\mathbf{a})$ is the real data pdf and $p_{gen}(\mathbf{a})$ is the pdf of the data generated by $G$.
The output of the optimal discriminator $D^*(\mathbf{a})$ can be used to estimate the pdf ratio 
$p_{gen}(\mathbf{a})/p_{data}(\mathbf{a}) = (1-D^*(\mathbf{a}))/D^*(\mathbf{a}).$

Now, the idea is to design a discriminator that estimates the probability density ratio $p_{XY}(\mathbf{x},\mathbf{y})/p_{Y}(\mathbf{y})$ and consequently the a-posteriori information as shown in \eqref{eq:PI}. 
The following Lemma provides an a-posteriori information estimator that exploits a discriminator trained to maximize a certain value function. At the equilibrium, a transformation of the discriminator output yields the searched density ratio.   

\begin{lemma}
\label{lemma:Lemma1}
Let $\mathbf{x}\sim p_X(\mathbf{x})$ and $\mathbf{y}\sim p_Y(\mathbf{y})$ be the channel input and output vectors, respectively. Let $H(\cdot)$ be the channel transfer function, in general stochastic, such that $\mathbf{y} = H(\mathbf{x}, \mathbf{h}, \mathbf{n} )$, with $\mathbf{h}$ and $\mathbf{n}$ being internal state and noise variables, respectively. 
Let the discriminator $D(\mathbf{x},\mathbf{y})$ be a scalar function of $(\mathbf{x},\mathbf{y})$.
If $\mathcal{J}_{MIND}(D)$ is the value function defined as
\begin{align}
\label{eq:discriminator_function}
\mathcal{J}_{MIND}(D) &= \; \mathbb{E}_{(\mathbf{x},\mathbf{y}) \sim p_{U}(\mathbf{x})p_{Y}(\mathbf{y})}\biggl[|\mathcal{T}_x|\log \biggl(D\bigl(\mathbf{x},\mathbf{y}\bigr)\biggr)\biggr] \nonumber \\ 
& + \mathbb{E}_{(\mathbf{x},\mathbf{y}) \sim p_{XY}(\mathbf{x},\mathbf{y})}\biggl[\log \biggl(1-D\bigl(\mathbf{x},\mathbf{y}\bigr)\biggr)\biggr],
\end{align}
where $p_U(\mathbf{x})=1/|\mathcal{T}_x|$ describes a multivariate uniform distribution over the support $\mathcal{T}_x$ of $p_X(\mathbf{x})$ having Lebesgue measure $|\mathcal{T}_x|$, then the optimal discriminator output is
\begin{equation}
\small
\label{eq:optimal_discriminator_1}
D^*(\mathbf{x},\mathbf{y}) = \arg \max_D \mathcal{J}_{MIND}(D) = \frac{p_{Y}(\mathbf{y})}{p_{Y}(\mathbf{y}) + p_{XY}(\mathbf{x},\mathbf{y})},
\end{equation}
and the a-posteriori information is computed as
\begin{equation}
\label{eq:i-DMIE}
i(\mathbf{x}|\mathbf{y}) = -\log_2 \biggl(\frac{1-D^*(\mathbf{x},\mathbf{y})}{D^*(\mathbf{x},\mathbf{y})}\biggr).
\end{equation}
\end{lemma}

\begin{proof}
From the hypothesis of the Lemma, the value function can be rewritten as
\begin{align}
\label{eq:Lebesgue1}
\mathcal{J}_{MIND}(D) &= \int_{\mathcal{T}_x} \int_{\mathcal{T}_y}\biggl[|\mathcal{T}_x|p_{U}(\mathbf{x})p_Y(\mathbf{y}) \log \biggl(D(\mathbf{x},\mathbf{y})\biggr) \nonumber \\ 
&+ p_{XY}(\mathbf{x},\mathbf{y}) \log \biggl(1-D(\mathbf{x},\mathbf{y})\biggr)\biggr] \diff \mathbf{x} \diff \mathbf{y}
\end{align}
where $\mathcal{T}_y$ is the support of $p_Y(\mathbf{y})$.
To maximize $\mathcal{J}_{MIND}(D)$, a necessary and sufficient condition requires to take the derivative of the integrand with respect to $D$ and equating it to $0$, yielding the following equation in $D$
\begin{equation}
\frac{|\mathcal{T}_x|p_{U}(\mathbf{x})p_Y(\mathbf{y})}{D(\mathbf{x},\mathbf{y})} -\frac{p_{XY}(\mathbf{x},\mathbf{y})}{1-D(\mathbf{x},\mathbf{y})} =0.
\end{equation}
The solution of the above equation yields the optimum discriminator $D^*(\mathbf{x},\mathbf{y})$ since: $p_U(\mathbf{x})=\frac{1}{|\mathcal{T}_x|}$ and $\mathcal{J}_{MIND}(D^*)$ is a maximum being the second derivative w.r.t. $D$ a non-positive function.

Finally, at the equilibrium, 
\begin{equation}
p_{X|Y}(\mathbf{x}|\mathbf{y})=\frac{p_{XY}(\mathbf{x},\mathbf{y})}{p_{Y}(\mathbf{y})} = \frac{1-D^*(\mathbf{x},\mathbf{y})}{D^*(\mathbf{x},\mathbf{y})},
\end{equation}
therefore the thesis follows.
\end{proof}

\section{Parametric (Neural Network) Implementation}
\label{sec:implementation}
The practical implementation of MIND is realized through a neural network-based learning approach. That is, the discriminator $D(\mathbf{x},\mathbf{y})$ with input $(\mathbf{x},\mathbf{y})$ is parameterized by a neural network. In a first phase, training of the discriminator is performed via optimization of \eqref{eq:discriminator_function}.
Training consists of transmitting repeatedly all coded messages such that at equilibrium the a-posteriori information $i(\mathbf{x}|\mathbf{y})$ is estimated. Then, in the testing phase, decoding can take place for each received new message $\mathbf{y}$ by minimizing the a-posteriori information. 
While the testing phase can be of moderate complexity depending on the neural network dimension, the training phase is more complex since no assumption on the channel and source statistics is made. In the following, we propose two different architectures and implementations of Lemma 1.

\subsection{Unsupervised Approach}

The value function in \eqref{eq:discriminator_function}, as it stands, requires the discriminator to be a scalar parameterized function which takes as input both the transmitted and received vector of samples $(\mathbf{x},\mathbf{y})$. The resulting architecture is presented in Fig.\ref{fig:unsupervised}. Such formulation is general and can be applied to any coding scheme. In fact, to learn the parameters of $D$, it is sufficient during training to alternate the input joint samples $(\mathbf{x},\mathbf{y})$ with marginal samples $(\mathbf{u},\mathbf{y})$, where $\mathbf{u}$ are realizations of a multivariate uniform distribution with independent components, defined over the support $\mathcal{T}_x$ of $p_X(\mathbf{x})$. Then, during the testing phase, decoding is accomplished by finding $\mathbf{x}$ that minimizes \eqref{eq:i-DMIE}, for all possible coded messages $\mathbf{x}$. This  method can be thought as an unsupervised learning approach since no known labels are exploited in the loss function for the training process. Hence, the objective is to implicitly estimate the density ratio in \eqref{eq:optimal_discriminator_1}.
Since in practical cases, the coded message $\mathbf{x}$ has discrete alphabet, another architecture is proposed below.

\begin{figure}[h]
	\centering
	\includegraphics[scale=0.43]{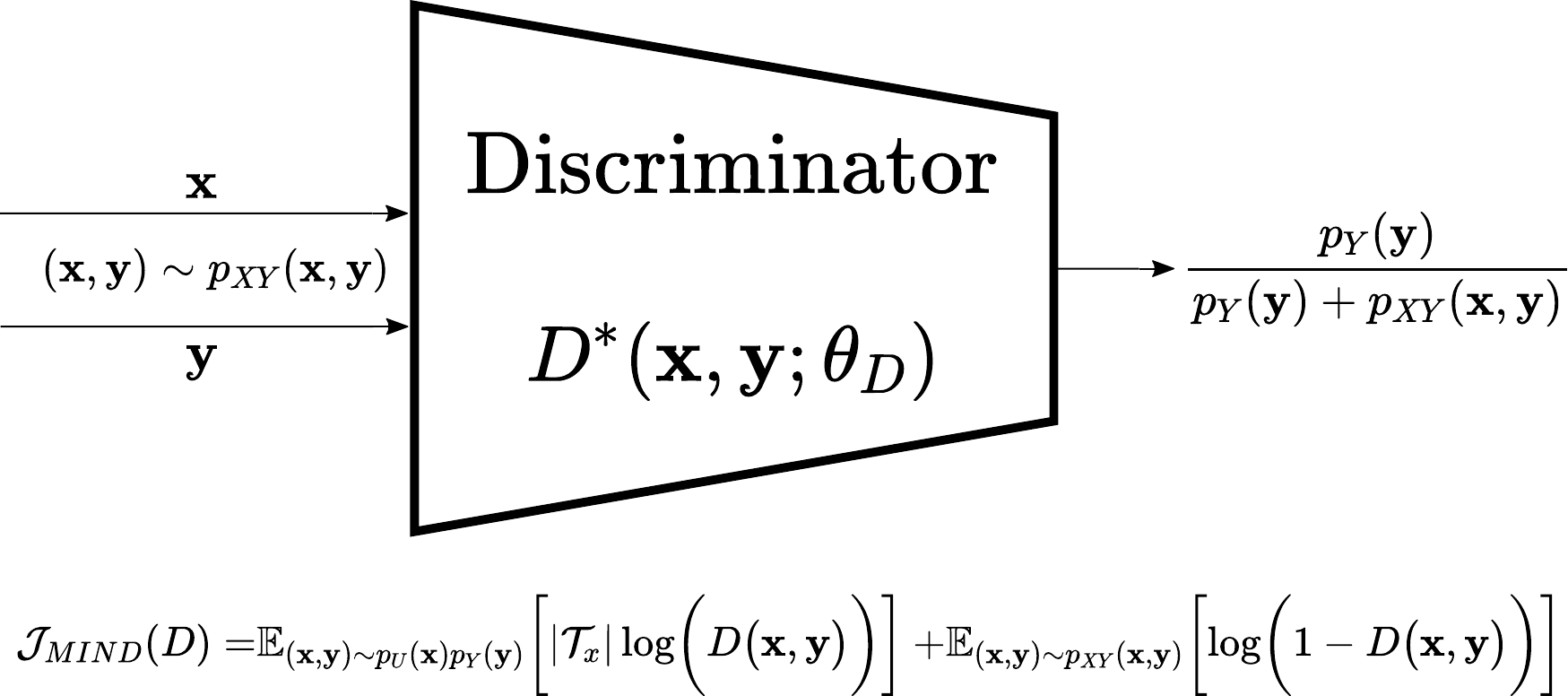}
	\caption{Unsupervised discriminator architecture and relative value function. Input is fed with paired samples from the joint distribution while the output consists of a single probability ratio value.}
	\label{fig:unsupervised}
\end{figure} 

\subsection{Supervised Approach}

It is possible to deploy a discriminator architecture that receives as input only the channel output $\mathbf{y}$. Indeed, if the channel input has discrete nature, e.g., $\mathbf{x}$ belongs to the alphabet $\mathcal{A}_x = \{\mathbf{x}_1, \mathbf{x}_2, \dots, \mathbf{x}_{M}\}$ with pdf $
    p_X(\mathbf{x})=\sum_{\mathbf{x}_i \in \mathcal{A}_x}{P_X(\mathbf{x}_i)\delta(\mathbf{x}-\mathbf{x}_i)},
$
then also the multivariate uniform pdf in Lemma 1 reads as $
    p_U(\mathbf{x})=\sum_{\mathbf{x}_i \in \mathcal{A}_x}{P_U(\mathbf{x}_i)\delta(\mathbf{x}-\mathbf{x}_i)},$
with $P_U(\mathbf{x}_i)=1/M$. Thus, the value function in \eqref{eq:discriminator_function}, or its Lebesgue integral expression in \eqref{eq:Lebesgue1}, can be decomposed as a sum of independent elements for every input codeword $\mathbf{x}_i$, for $i\in \{1,\dots,M\}$, as follows
\begin{equation}
\label{eq:j_i}
\mathcal{J}_{MIND}(D) = \sum_{i = 1}^{M}{\mathcal{J}_i(D)},
\end{equation}
where 
\begin{align}
\label{eq:discriminator_function_2}
\mathcal{J}_{i}(D) &= \; \mathbb{E}_{\mathbf{y} \sim p_{Y}(\mathbf{y})}\biggl[\log \biggl(D\bigl(\mathbf{x}_i,\mathbf{y}\bigr)\biggr)\biggr] \\ \nonumber
&+P_{X}(\mathbf{x}_i)\mathbb{E}_{\mathbf{y} \sim p_{Y|X}(\mathbf{y}|\mathbf{x}_i)}\biggl[\log \biggl(1-D\bigl(\mathbf{x}_i,\mathbf{y}\bigr)\biggr)\biggr].
\end{align}
Writing \eqref{eq:discriminator_function_2} with Lebesgue integrals and following the proof of Lemma 1, it is simple to verify that the discriminator which maximizes \eqref{eq:j_i} has to maximize all terms $\mathcal{J}_{i}(D)$, and this happens for 
\begin{equation}
\small
\label{eq:metric}
D^*(\mathbf{x}_i,\mathbf{y}) = \frac{p_{Y}(\mathbf{y})}{p_{Y}(\mathbf{y}) + P_{X}(\mathbf{x}_i)p_{Y|X}(\mathbf{y}|\mathbf{x}_i)} = \frac{1}{1 + P_{X|Y}(\mathbf{x}_i|\mathbf{y})},
\end{equation}
where $P_{X|Y}(\mathbf{x}_i|\mathbf{y})$ is the probability mass function of $X|Y$.
Now, it is more convenient from an implementation perspective to define an $M$-dimensional discriminator vector
\begin{equation}
    \bar{\mathbf{D}}(\mathbf{y}) := [D(\mathbf{x}_1,\mathbf{y}), \dots, D(\mathbf{x}_M,\mathbf{y})] = [D_1,\dots,D_M].
\end{equation}
In fact, the expectations in \eqref{eq:discriminator_function_2} can be rewritten as
\begin{align}
    &\mathcal{J}_{MIND}(D) =   \mathbb{E}_{\mathbf{y} \sim p_{Y}(\mathbf{y})}\biggl[\log \bigl(\bar{\mathbf{D}}(\mathbf{y}) \bigr)\cdot \bar{\mathbf{1}}_M^T\biggr] \\ \nonumber 
    &+\sum_{\mathbf{x}_i\in \mathcal{A}_x}{P_X(\mathbf{x}_i) \mathbb{E}_{\mathbf{y} \sim p_{Y|X}(\mathbf{y}|\mathbf{x}_i)}\biggl[\log \bigl(1-\bar{\mathbf{D}}(\mathbf{y})\bigr)\cdot \bar{\mathbf{1}}_M({\mathbf{x}_i})^T\biggr]},
\end{align}
where $\bar{\mathbf{1}}_M = [1,1,\dots,1]$ is a vector of $M$ unitary elements and $\bar{\mathbf{1}}_M({\mathbf{x}_i})$ is the one-hot code of $\mathbf{x}$ defined as
\begin{equation}
    \bar{\mathbf{1}}_M({\mathbf{x}_i}) := [0,\dots,0,\underbrace{1}_\text{i-th position},0,\dots,0].
\end{equation}
Finally, the supervised architecture (see Fig.\ref{fig:supervised}) value function assumes the following vector form 
\begin{align}
\label{eq:new_discriminator_function}
 &\mathcal{J}_{MIND}(D) =   \mathbb{E}_{\mathbf{y} \sim p_{Y}(\mathbf{y})}\biggl[\log \bigl(\bar{\mathbf{D}}(\mathbf{y}) \bigr)\cdot \bar{\mathbf{1}}_M^T\biggr] \\ \nonumber
 &+\mathbb{E}_{\mathbf{x} \sim p_{X}(\mathbf{x})} \mathbb{E}_{\mathbf{y} \sim p_{Y|X}(\mathbf{y}|\mathbf{x})}\biggl[\log \bigl(1-\bar{\mathbf{D}}(\mathbf{y})\bigr)\cdot \bar{\mathbf{1}}_M(\mathbf{x})^T\biggr].
\end{align}
The new expression \eqref{eq:new_discriminator_function} possesses the label information in the scalar product with the one-hot positional code and therefore can be treated as a supervised learning problem. 
\begin{figure}
	\centering
	\includegraphics[scale=0.43]{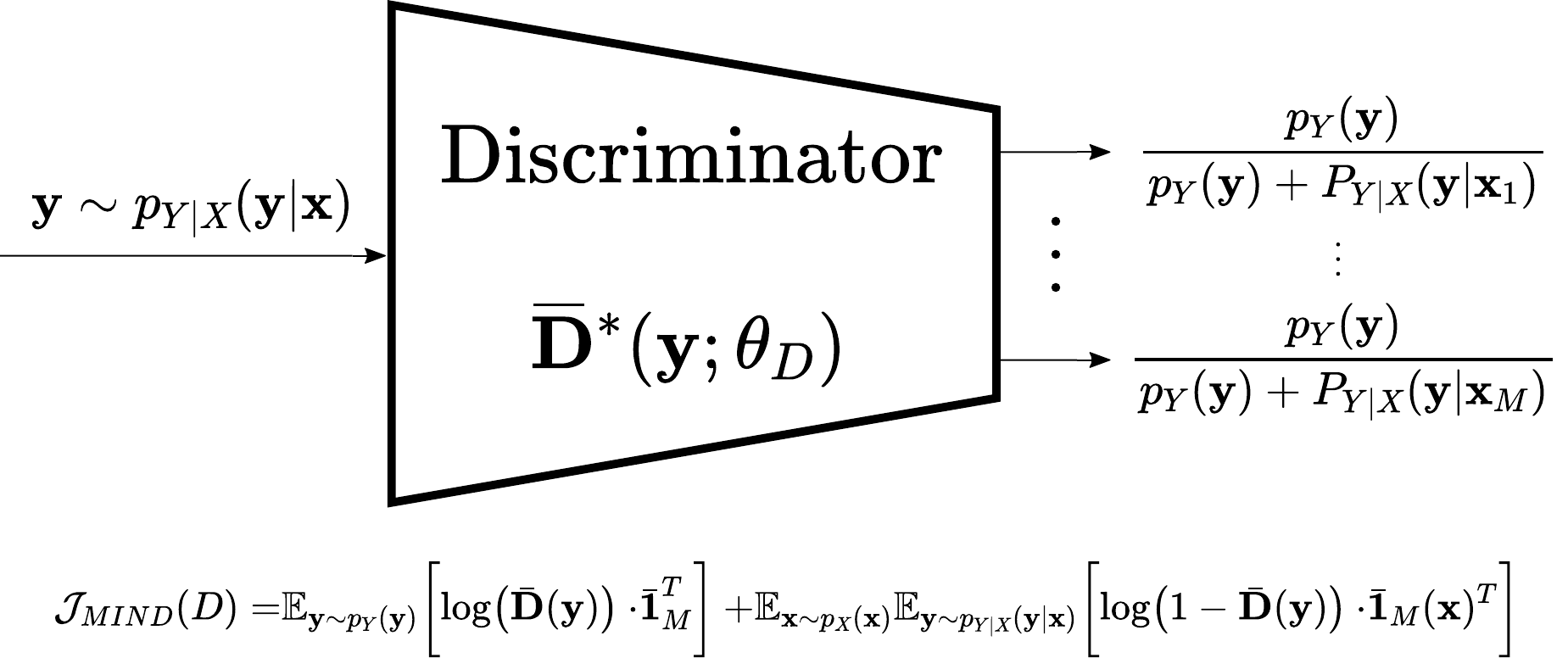}
	\caption{Supervised discriminator architecture and relative value function. Input is fed with the received channel samples while the output consists of multiple probability ratio values.}
	\label{fig:supervised}
\end{figure} 

\begin{figure*}
\begin{minipage}[t]{0.49\textwidth}
	\centering
  \includegraphics[scale=0.2]{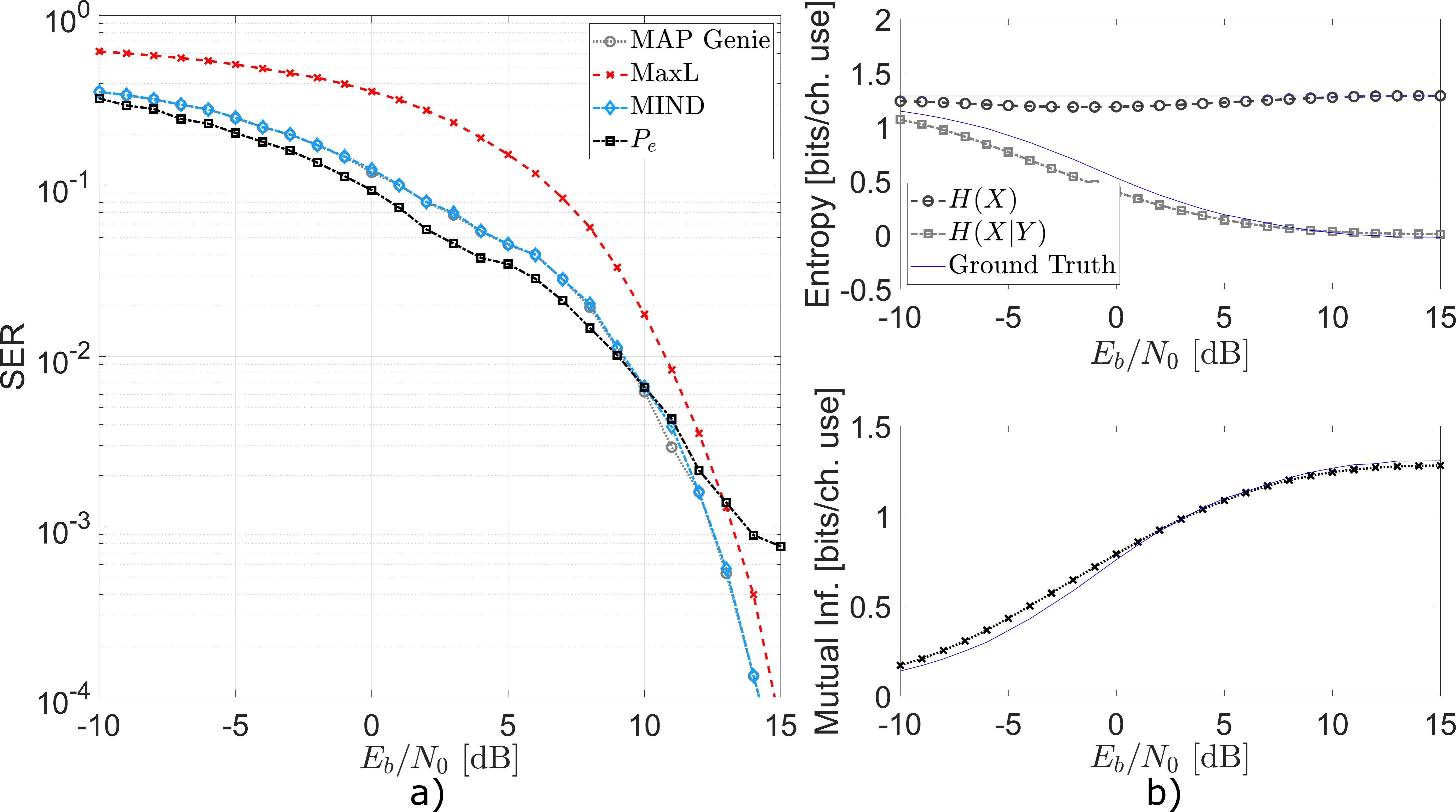}
  \caption{a) Symbol error rate for a 4-PAM modulation with non-uniform source distribution over an AWGN channel. Comparison among the optimal MAP decoder, MaxL decoder, MIND decoder and the estimated probability of error provided by MIND. b) Estimated source and conditional entropy (top) and estimated average mutual information (bottom) using MIND.}
	\label{fig:non-uniform_source}
\end{minipage}%
\hfill 
\begin{minipage}[t]{0.49\textwidth}
	\centering
   \includegraphics[scale=0.2]{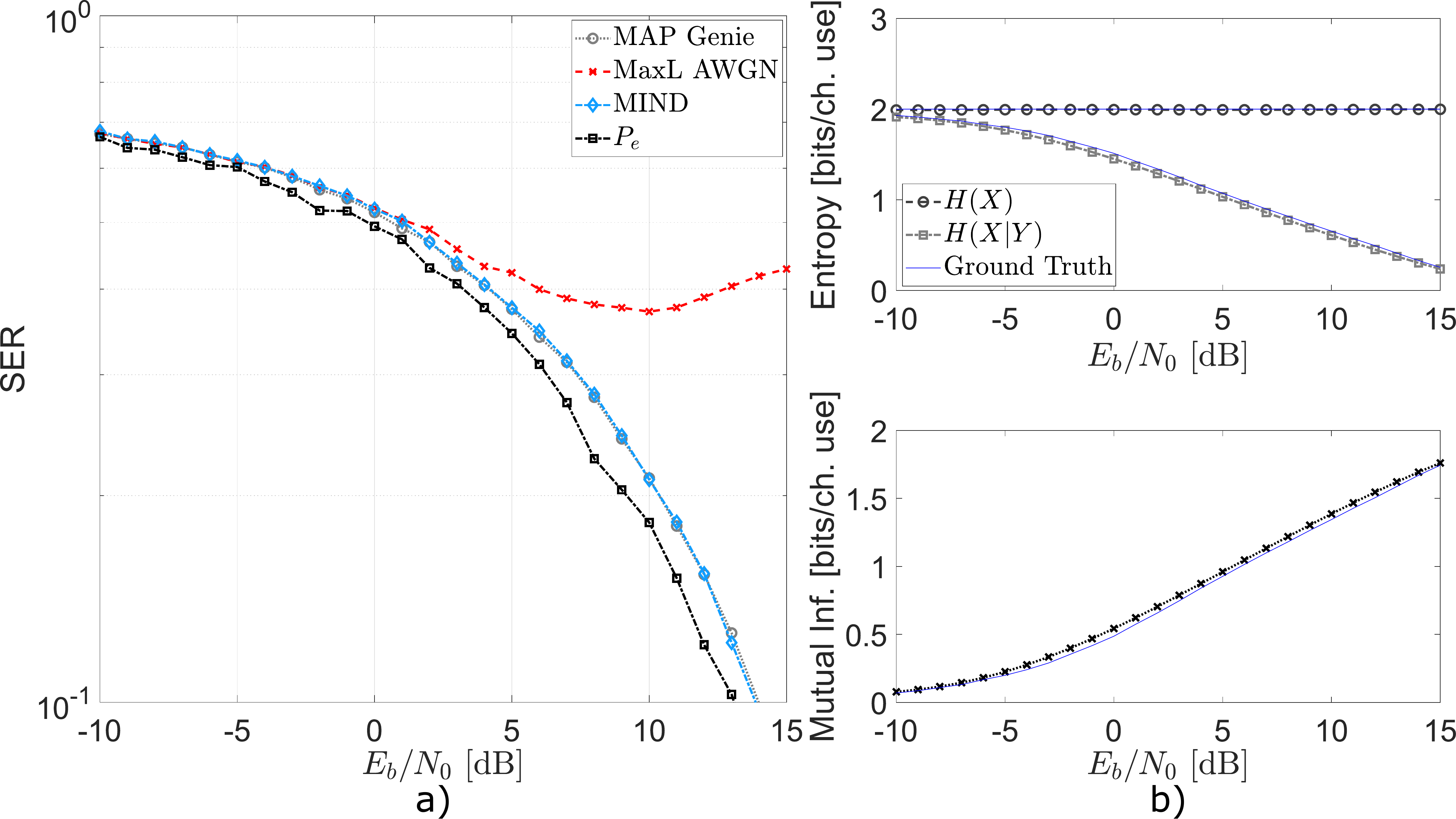}
  \caption{a) Symbol error rate for a 4-PAM modulation with uniform source distribution over a non-linear channel affected by additive Gaussian noise. Comparison among MaxL decoder with no CSI, MaxL decoder with perfect CSI, MIND decoder and the estimated probability of error. b) Estimated source and conditional entropy (top) and estimated average mutual information (bottom) using MIND.}
  \label{fig:non-linear_channel}
\end{minipage}%
\end{figure*}

By comparing the $M$ discriminator outputs $D^*_i = D^*(\mathbf{x}_i,\mathbf{y})$ in \eqref{eq:metric}, the final decoding stage implements
\begin{equation}
\mathbf{\hat{x}}_i = \argmax_{i \in \{1,...M\}} \log_2 \frac{1-D^*_i}{D^*_i} = \argmin_{\mathbf{x}_i \in \mathcal{A}_x} i(\mathbf{x}_i|\mathbf{y}),
\end{equation}
thus, MIND aims at finding the codeword $\mathbf{x}_i$ which minimizes the residual uncertainty $i(\mathbf{x}_i|\mathbf{y})$ after the observation of $\mathbf{y}$.

\section{Estimation of Achieved Information Rate and Decoding Error Probability}
\label{sec:estimation}
Dealing with an unknown channel, a relevant question is the estimation of the information rate achieved with the used coding scheme. MIND can be exploited for such a goal. In fact, the normalized average mutual information (in bits per channel use) is given by \cite{Gallager1968}
\begin{equation}
    I_n(X;Y) = \frac{1}{n}(H(X)-H(X|Y))
    \end{equation}
whose computation requires the entropy of the source $H(X)$ and the conditional entropy $H(X|Y)$:
\begin{equation}
\small
\label{eq:source_entropy}
H(X) = \mathbb{E}_{\mathbf{x} \sim p_X(\mathbf{x})}[i(\mathbf{x})] = -\sum_{\mathbf{x}_i \in \mathcal{A}_x}{P_X(\mathbf{x}_i)\log_2 P_X(\mathbf{x}_i)}
\end{equation}
\begin{equation}
\small
\label{eq:conditional_entropy}
H(X|Y) = -\mathbb{E}_{\mathbf{y} \sim p_Y(y)}\biggl[\sum_{\mathbf{x}_i \in \mathcal{A}_x}{P_{X|Y}(\mathbf{x}_i|\mathbf{y})\log_2 P_{X|Y}(\mathbf{x}_i|\mathbf{y})}\biggr].
\end{equation}
The discriminator in MIND, for a given coding scheme, returns an estimate of the a-posteriori probability mass function values $P_{X|Y}(\mathbf{x}_i|\mathbf{y})=(1-D_i)/D_i \; \forall i=\{1,\dots,M\}$, therefore \eqref{eq:conditional_entropy} can be directly computed. About \eqref{eq:source_entropy}, it is worth mentioning that the source distribution can be obtained using Monte Carlo integration from the a-posteriori probability obtained with MIND by averaging over $N$ realizations $\mathbf{y}_j$ of $Y$ as follows 
\begin{equation}
\small
\label{eq:monte-carlo_source}
P_X(\mathbf{x}_i) \stackrel{N\to \infty}{=} \frac{1}{N}\sum_{j=1}^{N}{P_{X|Y}(\mathbf{x}_i|\mathbf{y}_j)}.
\end{equation}
Similarly,
\begin{equation}
\small
H(X|Y) \stackrel{N\to \infty}{=} -\frac{1}{N} \sum_{j=1}^{N}{\sum_{\mathbf{x}_i \in \mathcal{A}_x}{P_{X|Y}(\mathbf{x}_i|\mathbf{y}_j)\log_2 P_{X|Y}(\mathbf{x}_i|\mathbf{y}_j)}}.
\end{equation}

%

MIND can also be used to estimate the decoding error probability $P_e$. Indeed, 
\begin{equation}
    P_e = 1-P[\mathbf{x}=\hat{\mathbf{x}}] = 1-\mathbb{E}_{\mathbf{y} \sim p_Y(y)}[P_{X|Y}(\mathbf{x}=\mathbf{\hat{x}}|\mathbf{y})],
\end{equation}
where the probability $P_{X|Y}(\mathbf{x}=\mathbf{\hat{x}}|\mathbf{y})$ comes from the decision criterion $\max_{\mathbf{x}_i} P_{X|Y}(\mathbf{x}_i|\mathbf{y})$ which is the instantaneous probability of correct decoding. Hence,
\begin{equation}
    P_e \stackrel{N\to \infty}{=} 1-\frac{1}{N} \sum_{j=1}^{N}{\max_{\mathbf{x}_i \in \mathcal{A}_x} P_{X|Y}(\mathbf{x}_i|\mathbf{y}_j)}.
\end{equation}

In the next section, we compare MIND to other well-known decoding criteria in scenarios for which an analytic solution is available. Furthermore, thanks to the network ability to estimate the a-posteriori probability, the average input-output mutual information is also reported.

\section{Numerical Results}
\label{sec:results}

The assessment of the proposed decoding approach is done by considering the following three representative scenarios: a) Uncoded transmission of symbols that are produced by a non-uniform source over an additive Gaussian noise channel; b) Uncoded transmission in a non-linear channel with additive Gaussian noise; c) Short block coded transmission in an additive Middleton noise channel.

Details about the architecture and hyper parameters of the neural networks are reported in the GitHub repository \cite{MIND_github}.

\subsection{Non-uniform Source}
To show the ability of the decoder to learn and exploit data dependence at the source, we firstly consider 4-PAM transmission with symbols $\mathbf{x}_i \in \{-3,-1,1,3\}$ and mass function given by a non-uniform source $p(\mathbf{x}) = [(1-P)/2, P/2, (1-P)/2, P/2]$, with $P=0.05$. It should be noted that data dependence at the source can be either created with a coding stage from a source that emits i.i.d. symbols, or intrinsically by the source of data traffic having certain statistics, as in this case. In Fig. \ref{fig:non-uniform_source}a the symbol error rate is shown. MIND is compared with the optimal MAP decoder that knows the a-priori distribution of the transmitted data symbols $p_X(\mathbf{x})$ and with the maximum likelihood (MaxL) decoder that only knows the Gaussian nature of the channel and assumes a source with uniform distribution. Moreover, the estimated $P_e$ is also shown. In Fig. \ref{fig:non-uniform_source}b, the source entropy, the conditional entropy and the average achieved mutual information, as estimated by MIND, are reported and compared with the ground truth numerically obtained. It is worth noticing that the average mutual information increases with the SNR, or alternatively, the residual uncertainty introduced by the channel decreases. In addition, the estimated source entropy tends to stabilize around its real value given by $-P\log_2(P/2)-(1-P)\log_2((1-P)/2)$.
\begin{figure*}[h]
	\centering
	\includegraphics[scale=0.25]{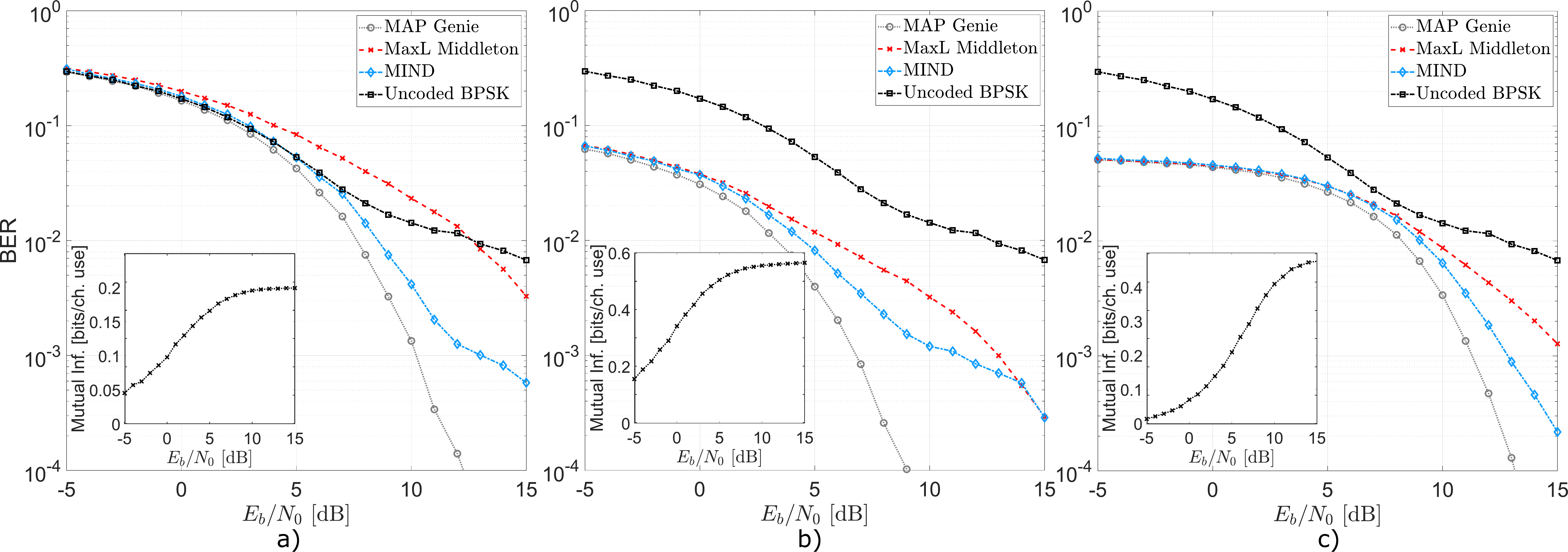}
	\caption{Bit error rate and mutual information of a short block coded transmission in an additive Middleton noise channel: a) Repetition code; b) Hamming code; c) Convolutional code.}
	\label{fig:middleton}
\end{figure*} 
\subsection{Non-linear Channel}
As a second example, we consider 4-PAM uncoded transmission with uniform source distribution over a non-linear channel with additive Gaussian noise. The objective of this experiment is to show the ability of MIND to discover such a channel non-linearity. In particular, the channel introduces a non-linearity (for instance because of the presence of non-linear amplifiers) modeled as $y_k=\text{sign}(x_k)\sqrt{|x_k|}+n_k$, where $k$ denotes the $k$-th time instant.
Fig.\ref{fig:non-linear_channel}a demonstrates how the MIND decoder manages to implicitly learn the non-linear channel model during the training phase and effectively use such information during decoding. A comparison with the MaxL decoder with and without perfect channel state information (CSI) knowledge is also conducted. Results show that MIND exhibits performance close to optimality. Fig. \ref{fig:non-linear_channel}b illustrates both the behaviour of the entropies and the average mutual information as function of the SNR.

\subsection{Additive Non-gaussian Noise Channel}
Lastly, and perhaps as most interesting example, we consider a short block coded transmission over an additive non-gaussian channel. The aim is to assess the ability of MIND to learn and exploit the presence of non-gaussian noise. In particular, we suppose that the noise follows a truncated Middleton distribution \cite{Middleton1977}, also called Bernoulli-Gaussian noise model, so that at any given time instant $k$ it is obtained as $n_k = (1-\epsilon_k)n_{1,k}+\epsilon_k n_{2,k}$ where $n_{1,k} \sim \mathcal{N}(0,\sigma_b^2)$ is a zero-mean Gaussian random variable with variance $\sigma_b^2$ and $n_{2,k} \sim \mathcal{N}(0,B\sigma_b^2)$ is also a zero-mean Gaussian random variable but with variance $B$ times larger. Instead, $\epsilon_k$ is a Bernoulli random variable with probability of success $P$. The pdf of the noise samples is then given by
$p_{N}(\mathbf{n}_k) = (1-P)\mathcal{N}(0,\sigma_b^2)+P\mathcal{N}(0,B\sigma_b^2).$
Assuming that the noise model is known and that BPSK symbols are transmitted, two decoding strategies can be devised. The first, denoted with MaxL Middleton, uses maximum likelihood decoding with the known conditional pdf $p(\mathbf{y}|\mathbf{x})=p_{N}(\mathbf{y}-\mathbf{x})$. The second strategy is a genie decoder that knows the outcome of the Bernoulli event for every time instance. That is, it knows whether a received sample is hit by Gaussian noise with variance $\sigma_b^2$ or $B\sigma_b^2$. A third decoding strategy is offered by MIND that learns the channel statistics.
We distinguish among three different types of codes: a) a binary repetition code with length $5$; b) a $(7,4)$ Hamming code; c) a rate $1/2$ convolutional code with memory $2$ and block-length $18$. For each of them, the bit error rate (BER) obtained with the genie, the MaxL Middleton and the MIND decoders is reported. The parameter $B$ was set to $5$ in all the experiments involving Middleton noise, wherein we also set $P=0.05$. 
Fig.\ref{fig:middleton} shows the gain provided by the neural-based decoder scheme over the classical maximum likelihood one that exploits the Middleton distribution. With MIND, the BER performance gets closer to the genie decoder. Fig.\ref{fig:middleton} also reports an estimate of the average mutual information provided by MIND.

\section{Conclusions}
\label{sec:conclusions}
In this paper, MIND, a neural decoder that uses the mutual information as decoding criterion, has been proposed. Two specific architectures have been described and they are capable of learning the a-posteriori information $i(\mathbf{x}|\mathbf{y})$ of the codeword $\mathbf{x}$ given the channel output observation $\mathbf{y}$ in unknown channels. Additionally, MIND allows the estimation of the achieved information rate with the used coding scheme as well as the decoding error probability. Several numerical results obtained in illustrative channels show that MIND can achieve the performance of the genie MAP decoder that perfectly knows the channel model and source distribution. It outperforms the conventional MaxL decoder that assumes the presence of Gaussian noise.        

\bibliographystyle{IEEEtran}
\bibliography{IEEEabrv,biblio}

\end{document}